\documentclass[12pt]{article}
 
\usepackage{amsmath}
\usepackage{theorem}
\usepackage{amssymb}
\usepackage{tikz}

\usepackage[osf,largesc,theoremfont]{newpxtext}
\usepackage[scaled=1]{inconsolata}
\usepackage[euler-digits,euler-hat-accent]{eulervm}
\usetikzlibrary{calc}
\usepackage[applemac]{inputenc}
\usetikzlibrary{positioning}
\usepackage{color}

\newcommand{\N}{\ensuremath{\mathbb{N}}}
\newcommand{\Z}{\ensuremath{\mathbb{Z}}}

\usepackage{theorem}

\newtheorem{theo_bis}{\noindent\bf Theorem}%
\newenvironment{theo}[2][]%
{\begin{theo_bis}{\bf #1} \sl #2}%
{\end{theo_bis}}

\newtheorem{propo_bis}{\noindent\bf Proposition}%
\newenvironment{propo}[2][]%
{\begin{propo_bis}{\bf #1} \sl #2}%
{\end{propo_bis}}

\newtheorem{lemma}{\noindent\bf Lemma}%
{\begin{lemma}{\bf #1} \sl #2}%
{\end{lemma}}

\newtheorem{r_bis}{\noindent\bf Remark}
{\begin{r_bis}{\bf #1} \normalfont #2}%
{\hfill $\Box$\end{r_bis}}

\newtheorem{e_bis}{\noindent\bf Example}
\newenvironment{ex}[2][]%
{\begin{e_bis}{\bf #1} \normalfont #2}%
{\hfill $\Box$\end{e_bis}}

\newtheorem{corollary}{\noindent\bf Corollary}%
{\begin{corollary}{\bf #1} \sl #2}%
{\end{corollay}}

\newenvironment{preuve}[2][]
{\begin{trivlist}{\bf\item Proof.} \textbf{#1} #2} 
{\hfill $\blacksquare$\end{trivlist}}

\long\def\symbolfootnote[#1]#2{\begingroup\def\thefootnote{\fnsymbol{footnote}}\footnote[#1]{#2}\endgroup} 

{\vspace{-0.2cm}\begin{list}%
	{--}%
	{\setlength{\labelwidth}{30pt}%
	 \setlength{\leftmargin}{35pt}%
	 \setlength{\itemsep}{-0.1cm}}}%
{ \end{list}\vspace{-0.2cm} }
\begin{document}

\title{Intergenerational geometric transfers of income\thanks{The second author gratefully acknowledges grant PID2023-146364NB-I00, funded by MCIU/AEI/10.13039/501100011033 and FSE+. This work is part of the R\&D\&I project grant PID2022-137211NB-I00, funded by MCIN/AEI/10.13039/501100011033/, the project ``ERDF A way of making Europe/EU'' and IMUS-María de Maeztu grant CEX2024-001517-M - Apoyo a Unidades de Excelencia María de Maeztu. Financial support from the Universit\'e Jean Monnet, Saint-Etienne, within the program ``Math\'ematiques de la d\'ecision pour l'ing\'enierie physique et sociale'' (MODMAD) is also acknowledged. Eric R\'emila and Philippe Solal would like to thank the hospitality received during their stay in the Universidad de Sevilla (IMUS and Department of Applied Mathematics II).}}

\author{Encarnaci\'on Algaba\thanks{Matem\'atica Aplicada II and Instituto de Matem\'aticas de la Universidad de Sevilla (IMUS), Escuela  Superior de Ingenieros, Camino 
de los Descubrimientos, s/n, 41092 Sevilla, Spain. E-mail: ealgaba@us.es} \\
\and Juan D. Moreno-Ternero\footnote{Department of Economics, Universidad Pablo de Olavide, 41013 Sevilla, Spain, E-mail: jdmoreno@upo.es} \\
\and Eric R\'emila\footnote{Universit\'e de Saint-Etienne, CNRS, GATE Lyon-St-Etienne UMR 5824, F-42023 Saint-Etienne, France. E-mail: eric.remila@univ-st-etienne.fr} \\ 
\and Philippe Solal\footnote{Universit\'e de Saint-Etienne, CNRS, GATE Lyon-St-Etienne UMR 5824, F-42023 Saint-Etienne, France, E-mail: solal@univ-st-etienne.fr} \\}

\maketitle

\begin{abstract}
\noindent We study intergenerational transfers of income. In our stylized model, each generation in an infinite (but countable) stream is endowed with some income. An allocation rule associates with each infinite stream another stream, thus involving intergenerational transfers of income. We single out a family of \textit{geometric} rules as a consequence of imposing axioms formalizing the principles of consistency, continuity and independence (as well as the basic requirements of feasibility and scale invariance). \\

\\

\end{abstract}
\newpage

\section{Introduction}
Economists, philosophers and social scientists alike have long been concerned with the issue of intergenerational equity. Koopmans (1960) and Diamond (1965) are the seminal contributions within the axiomatic literature on this issue. They both consider a stylized setting in which a social planner aims at ranking different (infinite but countable) streams of welfare, for a stream of generations. Diamond (1965) famously showed that impartiality cannot be combined with being sensitive to the interests of any single generation if the social preferences are required to be complete and continuous. Later contributions established related (and stronger) impossibility results (e.g., Asheim, 2010).
Many hold the view that concepts of welfare across persons are incommensurate (e.g., Rawls, 1971; Moreno-Ternero and Roemer, 2008). If one accepts that view, or simply faces the absence of comparable information about preferences, 
then one should depart from the above 
framework to deal with a more fundamental \textit{resourcist} framework in which each generation is associated with a level of resource (e.g., consumption, income, wealth) rather than welfare.\footnote{The term \textit{resourcist} is borrowed from Roemer (1986).} In such a setting, Epstein (1986) poses the dilemma that an economy has to choose between development and equity. 
Such a dilemma is resolved by Asheim (1991) for some domains of technologies, upon combining altruism with the exclusion of \textit{unjust} allocations. 
In a more general multi-commodity framework, Piacquadio (2014) 
also highlights the conflict between distributional equity and the Pareto principle.

In this paper, we focus on intergenerational transfers of income. That is, we also endorse a resourcist framework, in which we consider (infinite but countable) streams of income, for a stream of generations. In our framework, streams will be labelled according to the set of integer numbers $\mathbb{Z}$, instead of the set of non-negative integer numbers $\mathbb{N}$. That is, both the past and the future are unbounded. 
The aim is to design allocation rules that associate with each stream of income another stream of income, thus involving intergenerational transfers of income. 
Our approach will also be axiomatic and we shall single out a family of \textit{geometric} rules as a result of combining various axioms formalizing principles with normative appeal. On the one hand, we shall consider two basic and standard operational axioms: \textit{feasibility} (allocation rules cannot yield streams with a higher aggregate income than the original stream) and \textit{scale invariance} (which implies that the units in which income is measured are irrelevant). On the other hand, we consider axioms formalizing the principles of independence, consistency and continuity. {\it Independence of a future income} states that changing the income of one generation does not affect the allocation of past generations. {\it Consistency} refers to the scenario in which the allocation is reevaluated in a given generation, assuming past generations have zero income (with the interpretation that they already ``left'' with their awards) and the present generation's income is augmented with the residual income from past generations (the intergenerational transfers they gave), whereas the future generations keep their original incomes. The axiom states that the allocation for the present and future generations should not change in this new scenario. Finally, {\it continuity} is the standard notion formalizing that small changes in the original income stream can only produce small changes in the income stream the rule yields. Nevertheless, various continuity notions can be considered, depending on the norm at stake. In this case, we endorse the standard 1-norm (or taxicab norm). We shall also explore alternative continuity notions arising from the infinity norm and the point-wise norm. A side result of our analysis is that different continuity notions will have different implications, when combined with the rest of the axioms in our analysis. 

Our work is also connected to other strands of the axiomatic literature that have evolved recently. 

Income redistribution problems for finite generations were formally introduced by Ju et al. (2007), as an instance of their generalized claims problems.\footnote{Their model for income redistribution is thus similar to the taxation model studied by Young (1988), which reinterprets the seminal model to study the problem of adjudicating conflicting claims introduced by O'Neill (1982).} They focus on the notion of \textit{reallocation-proofness} (a group of agents cannot manipulate the outcome of the rule upon reallocating the incomes within the group) and show that this axiom, together with \textit{no transfer paradox} (transferring some income before redistribution takes place does not increase income after redistribution) characterizes the family of income-tax schedules with a flat tax rate and personalized lump-sum transfers. Chambers and Moreno-Ternero (2021) studied the bilateral case of this model and showed that \textit{continuity}, 
\textit{no transfer paradox} 
and \textit{idempotency} characterize a large family (dubbed \textit{threshold} rules), that guarantee partial redistribution for unequal incomes. Mart\'{\i}nez and Moreno-Ternero (2022) show that the combination of \textit{equal treatment of equals}, \textit{additivity}, and \textit{idempotency} characterizes a pair of rules: \textit{full redistribution} and \textit{laissez-faire} (the counterpart of the no-transfer rule in our setting). Dismissing the latter axiom, one widens the scope to characterize a family of linear combinations between both rules. A similar family was characterized earlier by Casajus (2015), resorting to a \textit{monotonicity} axiom instead of \textit{additivity}.  

From a different vantage point, Hougaard et al. (2017) studied the problem of revenue allocation in a hierarchy. 
Therein, Hougaard et al. (2017) characterize (uniform) geometric rules by the combination of {\sl lowest rank consistency}, {\sl highest rank revenue independence}, {\sl highest rank splitting neutrality}, and {\sl scale invariance}.\footnote{Harless (2020) and Ju et al. (2024) present related results for the same model. Dietzenbacher and Kondratev (2023) and Kondratev et al. (2024) also characterize geometric rules in the related setting of prize allocation in competitions.} The first axiom in the list, albeit similar, is different from the consistency axiom we use in this setting, as it is a variable-population axiom assuming that an agent in the hierarchy (the one at the bottom, a concept that cannot be defined in our setting) departs. Likewise, the second one in the list could be seen as a counterpart of our independence axiom (although, again, we cannot refer to highest rank because our streams are infinite). Finally, the {\sl highest rank splitting neutrality} axiom is akin to the notion of reallocation-proofness mentioned above, but in a variable-population setting too and, thus, it does not have a counterpart in our setting. 
Hougaard et al. (2022) and Gudmundsson et al. (2025) also consider geometric rules for hierarchies that are generated endogenously and, thus, can be infinite too. But, in that case, the domain of streams would be the set of non-negative integer numbers. Furthermore, the focus is on game induced by allocation rules that distribute the total value generated among agents. The main message of their analysis is that relatively simple rules that prioritize short-term incentives can effectively achieve long-term systemic goals.

Finally, Mart\'{\i}nez and Moreno-Ternero (2025, 2026) study the allocation of riparian rights with a linear structure. This is a resourcist version of the river sharing model introduced by Ambec and Sprumont (2002). And it can also be seen as a mirror image of the model in Hougaard et al. (2017) mentioned above, for linear hierarchies. Mart\'{\i}nez and Moreno-Ternero (2026) characterize in this setting a family of geometric rules combining three axioms: \textit{partial-implementation invariance}, \textit{upstream invariance}, and \textit{scale invariance}. As feasibility is included in the definition of rules in that setting, this result can be seen as a counterpart of our main theorem. The difference between both results essentially lies in the \textit{continuity} axiom we consider here, which plays a crucial role in the analysis of infinite streams, as we shall discuss later in the text. 


\section{The model}
Assume that the set of integers $\mathbb{Z}$ represents the infinite (but countable) set of generations.\footnote{Formally, $\mathbb{Z} = \{\ldots, - 1, 0, 1, \ldots \}$. The symbol $0_{\mathbb{Z}}$ stands for the null stream, i.e., $(0_{\mathbb{Z}})_i = 0$ for each $i \in \mathbb{Z}$.
For each $i \in \mathbb{Z}$, $e^i: \mathbb{Z} \longrightarrow \mathbb{R}_+$ is defined as $e_i^i = 1$ and $e^i_j = 0$ for $j \in \mathbb{Z} \setminus i$. Given $x \in \mathbb{R}$, $| x|$ stands for the absolute value of $x$. } 
The natural interpretation is that $0$ refers to the present generation, whereas negative numbers refer to past generations and positive numbers to future generations. Let $r_i \in \mathbb{R}_{+}$ be the income that generation $i \in \mathbb{Z}$ generates. Denote by $L^1$ the set of functions $r: \mathbb{Z} \longrightarrow \mathbb{R}$, $i \longmapsto r_i $ such that $\sum_{i \in \mathbb{Z}} | r_i | < +\infty$, and by $L_+^1$ the subset of functions $r$ in $L^1$ such that $r_i \in \mathbb{R}_+$ for each $i \in \mathbb{Z}$. An {\bf infinite stream of income} is $r \in  L_+^1$. An {\bf allocation rule} is a function $\phi: L_+^1 \longrightarrow L_+^1$. 
\medskip

The following family of allocation rules, which we shall denote by ${\cal G}^{\lambda}$, will be central in our analysis. 
\medskip

\noindent {\bf Geometric rules  ($\phi^{ \lambda}$)}: There is a function $\lambda: \mathbb{Z} \longrightarrow [0, 1]$,    $i \longmapsto \lambda_i$ such that for each $r \in L^1_+$, and each $i \in \mathbb{Z}$, 
\begin{equation} \label{GR}
\phi_i^{ \lambda} (r)  =   \lambda_i\sum_{j\in \mathbb{Z}:\atop j \leq i} \prod_{k = j}^{i - 1}(1- \lambda_k)r_j.
\end{equation}

A geometric rule represents income transfers from past generations to future generations.
More precisely, the quantity 
$$\prod_{k = j}^{i - 1}(1- \lambda_k)r_j$$ 
is the income transfer from generation $j$ to generation $i$. Each generation $k$ between $j$ and $i-1$ retains a share $\lambda_k$ of the income received from $j$ and transfers the remaining share to the next generation $k+1$. It follows that
  $$\sum_{j\in \mathbb{Z}:\atop j < i} \prod_{k = j}^{i - 1}(1- \lambda_k)r_j$$ 
  is the income transferred by generations $j < i$ to generation i.
The total income of generation $i$ is thus given by:
 $$r_i + \sum_{j\in \mathbb{Z}:\atop j < i} \prod_{k = j}^{i - 1}(1- \lambda_k)r_j.$$ 
 Generation $i$ retains a share $\lambda_i$ of this income and transfers the remaining share to the next generation $i +1$.

\medskip


Suppose $e^0$ is the initial stream of income (that is, the present generation has one unit of income, whereas all other generations have zero). Then, for each $i < 0$, $\phi_i^{\lambda} (e^0) = 0$ and $\phi_0^{\lambda} (e^0) = \lambda_0$. That is, generation $0$ keeps $\lambda_0$ of its unit income and transfers $1 - \lambda_0$ of it 
to the next generation $1$, whose income becomes $1 - \lambda_0$. Therefore, $\phi_1^{\lambda} (e^0) = \lambda_1( 1 - \lambda_0)$. Generation 1 transfers $(1 - \lambda_1)$ of its income to generation 2 whose income becomes  $(1 - \lambda_1) (1 - \lambda_0)$ and so its allocation is given by $\phi_2^{\lambda} (e^0) = \lambda_2(1 - \lambda_1) (1 - \lambda_0)$. Proceeding recursively, one obtains that, for each $i \geq 1$,
$$\phi_i^{\lambda} (e^0) = \lambda_i \prod_{j = 0}^{i - 1} ( 1 - \lambda_j).$$

A {\bf geometric rule} is {\bf uniform} if $\lambda: \mathbb{Z} \longrightarrow [0, 1]$ is a constant function. In other words, with a slight abuse of notation, there exists $\lambda \in [0, 1]$ such that, for each $i \in \mathbb{Z}$, $\lambda_i = \lambda$. Thus, for each $r \in L^1_+$, and each $i \in \mathbb{Z}$, 
\begin{equation}
\label{GR-U}
\phi_i^{ \lambda} (r) 
= \lambda \sum_{j \leq i} (1 - \lambda)^{i - j} r_j.
\end{equation}

We shall denote the resulting family by ${\cal U}^{\lambda}\subset {\cal G}^{\lambda}$.\\

Consider again $e^0$ as the initial stream of income. Then, the uniform geometric rule $\phi^{\lambda}$ yields 
\begin{equation*}
\phi_i^{\lambda} (e^0) =\left\{ 
\begin{tabular}{cc}
$ \lambda ( 1 - \lambda)^i$ & if $i \geq 1 $ \\ 
$\lambda$ & if $i = 0 $\\
$0$ & otherwise.%
\end{tabular}%
\right.
\end{equation*}

Two extreme and polar uniform geometric rules emerge when either $\lambda =0$ or $\lambda =1$. In the former case, each generation transfers all its income to the next generation. In the latter, each generation transfers nothing to the next generation. We shall refer to them as the full-transfer and no-transfer rules, respectively.\footnote{The name \textit{full transfer} refers to the fact that each generation fully transfers its income, albeit no generation receives a transfer. This is a stark contrast with the finite case, in which the last agent receives the aggregate income (from all the transfers).} Formally,\\

\noindent {\bf Full-Transfer rule ($\phi^{0}$)}: for each $r \in L^1_+$,
\begin{equation*} \label{FT}
\phi^{0} (r) = 0_{\mathbb{Z}}.
\end{equation*}

\noindent {\bf No-Transfer rule ($\phi^{1}$)}: for each $r \in L^1_+$,
\begin{equation*} \label{NT}
\phi^{1} (r) = r.
\end{equation*}

\section{Axiomatic analysis}
\subsection{Axioms}
We formalize some principles with normative appeal for allocation rules in this setting. \\

\noindent First, the axiom stating that allocation rules cannot yield streams with higher aggregate income.\medskip 

\noindent {\bf Feasibility}: For each $r \in L^1_+$, 
$$\sum_{i \in \mathbb{Z}} \phi_i (r) \leq \sum_{i \in \mathbb{Z}} r_i. $$

\noindent Second, the strengthening of the previous axiom indicating that allocation rules cannot waste income in the process either. \medskip

\noindent {\bf Balance}: For each $r \in L^1_+$, 
$$\sum_{i \in \mathbb{Z}} \phi_i (r) = \sum_{i \in \mathbb{Z}} r_i. $$


\noindent Next, a standard axiom implying that the units in which we measure incomes is irrelevant.\medskip 

\noindent {\bf Scale invariance}: For each $r \in L^1_+$, and each $\alpha \in \mathbb{R}_+$, 
$$\phi(\alpha r) =  \alpha \phi (r). $$

\noindent We now move to independence axioms, with respect to a specific future generation, or (more strongly) a whole (future) stream.\footnote{These axioms are somewhat reminiscent of the classical notions of intertemporal myopia (e.g., Brown and Lewis, 1981). } \medskip

\noindent {\bf Independence of a future income}: 
Let $r,r'\in L^1_+$ be such that $r_i = r'_i$ for each $i \in \mathbb{Z} \setminus j$. Then, for each $i < j$,
$$\quad \phi_i (r) = \phi_i (r').$$ 
\medskip 

\noindent {\bf Independence of future income streams}: 
Let $r,r'\in L^1_+$ be such that $r_i = r'_i$ for each $i < j$. Then, for each $i < j$,
$$\phi_i (r) = \phi_i (r').$$ 
\medskip

\noindent For the next axiom, we need to introduce a preliminary notion first. 
A generation $i \in \mathbb{Z}$ is {\bf inessential} in $r \in L_+^{1}$ if for each $j \leq i$, $r_j = 0$.\footnote{This is inspired by the concept of inessential player in TU-games with a permission structure (e.g., van den Brink and Gilles, 1996).}\medskip

\noindent {\bf Zero income for inessential generations}: 
For each inessential generation $i$ in 
$r\in L^1_{+}$, 
$$\phi_i (r) = 0. $$

\noindent We now consider an invariance notion referring to translations of incomes.\medskip 

\noindent {\bf Translation invariance}: 
For each pair $r, r' \in L^1_+$ such that $r_{i +1}' = r_i$ for each 
$i \in \mathbb{Z}$, 
$$ \phi_{i} (r) = \phi_{i+1} (r').$$

\noindent The next axiom is a sort of stability condition indicating that allocations cannot be revisited.\medskip

\noindent {\bf Idempotency}: For each $r \in L^1_+$, 
$$\phi (\phi(r)) = \phi(r). $$

\noindent Another stability condition comes next.\footnote{The axiom aligns with the partial implementation axiom introduced by Dietzenbacher et al. (2024) for claims problems. For other notions of consistency and some of the applications this principle has in axiomatic work, the reader is referred to Thomson (2012). See also Moreno-Ternero and Roemer (2006).} 
Given an allocation rule $\phi$, a stream of income $r \in L^1_+$ and a generation $j \in \mathbb{Z}$, define the stream of income $r^j$ as follows:
\begin{equation*}
r_i^j =\left\{ 
\begin{tabular}{cl}
$r_i$ & if $i > j $ \\ 
$r_j + \sum_{i \in \mathbb{Z}: \atop i < j} \bigl(r_i - \phi_i (r) \bigr)$ & if $i = j $ \\
$0$ & otherwise.%
\end{tabular}%
\right.
\end{equation*}
That is, $r^j$ is obtained by assuming that, after the allocation takes place, $j$ adds to its income the residual incomes from all previous generations, whereas those previous generations appear with zero income, and future generations remain with their original incomes. 
Note that $r^j$ does not necessarily belong to $L_+^1$. But if it does, the next requirement says future generations should receive the same in the allocation for this new stream and in the allocation for the previous stream. \\

\noindent {\bf Consistency}: 
For each $r \in L^1_+$ and each $j \in \mathbb{Z}$ such that $r^j \in L^1_+$, 
$$\phi_{i} (r) = \phi_{i} (r^j),$$
for each $i \geq j$. \medskip

\noindent We conclude this inventory of axioms with a basic regularity notion.\footnote{We shall discuss alternative notions of continuity later in the text.}\medskip

Let $(r^m)_{m \geq 0}$ be a sequence of elements of $ L_+^1$. We say that $(r^m)_{m \geq0}$ converges to $r \in  L_+^1$ in the taxicab norm, which we write as $(r^m)_{m \geq0} \overset{1}\to r$, if for each $\epsilon >0$, there exists $m_0 \in \N$, such that for each $m \geq m_0$,  $\sum_{i \in \Z}\vert r^m_i  -r_i \vert  < \epsilon$.  \\

\noindent {\bf Continuity}: If $(r^m)_{m \geq0}\overset{1}\to r $, then $(\phi(r^m))_{m \in \N}\overset{1}\to \phi(r)$.
\subsection{Preliminary results}
Let ${\cal B}^{\lambda}$ denote the sub-family of geometric rules satisfying the following condition:
\begin{equation} \label{c-balance}
\forall i \in \mathbb{Z}, \quad \prod_{k = i }^{+\infty} ( 1 - \lambda_k) = 0.
\end{equation}
To understand (\ref{c-balance}), pick any generation $i$. It transfers $1 - \lambda_i$ of its income $r_i$ to generation $i+1$. Then, generation $i+1$ transfers  $1 - \lambda_{i +1}$ of $1 - \lambda_i$ of $r_i$ to generation $i +2$. Therefore,   $\prod_{k = i }^{ i + p} ( 1 - \lambda_k)$, $p \geq 1$, represents the share of $r_i$ obtained by generation $i + p$. Condition (\ref{c-balance}) 
states that generations distant from $i$ receive an increasingly smaller share of the income $r_i$ and that, ultimately, transfers of $r_i$ no longer take place.
In other words, income $r_i$  is fully used by $i$ or transferred to generations $i+1, i+2, \ldots$ No share of this income is transferred indefinitely without being used for  transfer to a generation. 

Let ${\cal E}^{\lambda}$ denote the sub-family of geometric rules satisfying the following condition: 
 there exists $\epsilon \in (0, 1)$ such that $\lambda_k > \epsilon$ for each $k \in \mathbb{Z}$.  
It follows that ${\cal E}^{\lambda} \subset{\cal B}^{\lambda}$. To see this, note that, for each $k \in \mathbb{Z}$, $(1 - \lambda_k) < 1 - \epsilon$. Thus, $ \lim_{p \to + \infty}(1 - \epsilon)^p = 0,$ from which one deduces that (\ref{c-balance}) holds. 

\begin{ex}
Let $\lambda: \mathbb{Z} \longrightarrow [0, 1]$ be defined as:
\begin{equation*}
\lambda_i=\left\{ 
\begin{tabular}{cc}
$1 - \exp \biggl(- \frac{1}{2^i}\biggr)$ & if $i > 0 $ \\ 
$0$ & otherwise.%
\end{tabular}%
\right.
\end{equation*}
Consequently,
$$\prod_{k = 1}^{+ \infty} ( 1 - \lambda_k) = \exp \biggl( \sum_{k = 1}^{+ \infty} - \frac{1}{2^k}\biggr) =\exp ( - 1).$$
More generally, for each $i > 0$,  
$$\prod_{k = i}^{+ \infty} ( 1 - \lambda_k) =\exp \biggl( - \frac{1}{2^{i -1}} \biggr).$$
Thus, if we consider the corresponding geometric rule $\phi_i^{\lambda}$, it follows that 
$$\sum_{i \in \mathbb{Z}:\atop i < j} \phi_i^{\lambda} (r) =   \sum_{i \in \mathbb{Z}} r_i - \sum_{i > 0} \exp \biggl( - \frac{1}{2^{i -1}} \biggr) r_i,$$
from which one concludes that $\phi_i^{\lambda}$ does not satisfy balance. 
\end{ex}

The geometric rule in the previous example does not belong to ${\cal B}^{\lambda}$. 
As the next result states, all geometric rules satisfy feasibility, but only those within ${\cal B}^{\lambda}$ satisfy balance.

\begin{propo} \label{propBalance} The following statements hold:
  \begin{itemize}
\item All geometric rules satisfy feasibility.
\item A geometric rule 
satisfies balance if and only if it belongs to ${\cal B}^{\lambda}$. 
\end{itemize}
\end{propo}


The next results state that all geometric rules satisfy continuity and that this, together with the previous result, will have further implications for some of the remaining axioms listed above. 

\begin{propo} \label{Lip-continuous} All geometric rules satisfy continuity. 
\end{propo}

\begin{propo} \label{relationaxiom} 
  The following statements hold:
\begin{enumerate}
\item 
Independence of a future income and continuity imply independence of future income streams; 
\item  Independence of future income streams and feasibility imply zero income for inessential generations.
\end{enumerate}
\end{propo}

The proofs of the previous propositions are in the appendix. We are now ready to state our main characterization result.

\subsection{The main result}

\begin{theo} \label{charaResult} An allocation rule satisfies feasibility, scale invariance, independence of a future income, consistency and continuity if and only if it is a geometric rule. 
\end{theo}
\begin{preuve} 
By Proposition \ref{propBalance} and Proposition \ref{Lip-continuous}, all geometric rules satisfy \textit{feasibility} and \textit{continuity}. It is straightforward to show that they also satisfy \textit{scale invariance} and \textit{independence of a future income}. As for \textit{consistency}, let $r \in L^1_+$ and $j \in \mathbb{Z}$. Then,
\begin{eqnarray}
r_j^j &=& r_j + \sum_{\ell \in \mathbb{Z}: \atop \ell < j}  (r_{\ell} - \phi_{\ell} ^{\lambda} (r)) \nonumber \\
       & =& r_j + \sum_{i \in \mathbb{Z}: \atop \ell < j}  r_{\ell} - \sum_{\ell \in \mathbb{Z}:\atop \ell < j} \bigl( 1 -   \prod_{k = \ell}^{j - 1}  ( 1 - \lambda_k)\bigr)r_{\ell} \nonumber \\
      & = &  r_j + \sum_{\ell \in \mathbb{Z}:\atop \ell < j} \prod_{k = \ell}^{j - 1}  ( 1 - \lambda_k)r_{\ell} \nonumber \\
      &= &\sum_{\ell \in \mathbb{Z}:\atop \ell \leq j} \prod_{k = \ell}^{j - 1}  ( 1 - \lambda_k)r_{\ell}, \nonumber
\end{eqnarray} 
where the second equality comes from Lemma \ref{Sum} (in the appendix). Note that $r^j \in L^1_+$. Therefore, using the definition of $r^j$ and the above equality, one obtains that, for each $i \geq j$, 
\begin{eqnarray}
\phi_i^{\lambda}(r^j) & = & \lambda_i \biggl(\sum_{\ell \in \mathbb{Z}: \atop \ell \leq i} \prod_{k = \ell}^{i - 1} ( 1 - \lambda_k)r_{\ell}^{j} \biggr) \nonumber \\
                                                                  & = & \lambda_i \biggl(\sum_{\ell \in \mathbb{Z}: \atop j < \ell \leq i} \prod_{k = \ell}^{i - 1} ( 1 - \lambda_k)r_{\ell}^{} +  \prod_{k = j}^{i - 1}  ( 1 - \lambda_k)r_{j}^j \biggr) \nonumber \\
                                                                  &= &  \lambda_i \biggl(\sum_{\ell \in \mathbb{Z}: \atop j < \ell \leq i} \prod_{k = \ell}^{i - 1} ( 1 - \lambda_k)r_{\ell} +\prod_{k = j}^{i - 1}  ( 1 - \lambda_k)
                                                                 \bigl (\sum_{\ell \in \mathbb{Z}:\atop \ell \leq j} \prod_{k = \ell}^{j - 1}  ( 1 - \lambda_k)r_{\ell} \bigl) \biggr) \nonumber \\
                                                                 &= &\lambda_i \biggl(\sum_{\ell \in \mathbb{Z}: \atop j < \ell \leq i} \prod_{k = \ell}^{i - 1} ( 1 - \lambda_k)r_{\ell} + \sum_{\ell \in \mathbb{Z}:\atop \ell \leq j} \prod_{k = \ell}^{i- 1}  ( 1 - \lambda_k)r_{\ell}  \biggr) \nonumber \\
 &= & \phi_i^{\lambda}  (r), \nonumber                                                               
\end{eqnarray}
as desired. 
\medskip

Conversely, let $\phi$ be an allocation rule satisfying all the axioms in the statement. For each $i \in \mathbb{Z}$, let $\lambda_i=\phi_i(e^i)$. By feasibility, $\lambda_i \in [0, 1]$. Let $\lambda: \mathbb{Z} \longrightarrow [0, 1]$ be the corresponding stream. We show that $\phi = \phi^{\lambda}$. 
To do so, we proceed in several steps.  
\medskip

\noindent \text{\bf Step 1}. For each $i \in \mathbb{Z}$, $\phi_i(e^i) = \phi_i^{\lambda}(e^i)$. It trivially follows from the above. 
\medskip

\noindent \text{\bf Step 2}. For each $r \in L^1_+$ such that there exists $j \in \mathbb{Z}$ where $r_i  = 0$ for each $i \leq j$, $\phi_i(r) = \phi^{\lambda}(r)$. \medskip

Let $r$ be one of those streams. By Proposition \ref{relationaxiom}, $\phi$ satisfies \textit{zero income for inessential generation}. Then, for each $i \leq j$,
$$\phi_{i} (r) = 0 = \phi_i^{\lambda}(r).$$
For each $i > j$, we proceed by induction on $i - j \geq 1$. 

\noindent {\sc Base Case}: $i - j = 1$. That is, $i = j +1$. By Proposition \ref{relationaxiom}, $\phi$ satisfies \textit{independence of future income streams}. Then,
$$\phi_i (r) = \phi_i (r_i e_i).$$
By Step 1 and scale invariance,
$$ \phi_i (r) =\phi_i (r_i e_i) = r_i \lambda_i = \phi_i^{\lambda}(r),$$
as desired.

\noindent {\sc Induction hypothesis}. Assume that it holds for each generation $\ell$ such that $j < \ell \leq i - 1$ and $i - 1 - j \geq 1$.

\noindent {\sc Induction step}. Consider generation $i$. By the induction hypothesis, $ \phi^{}_{\ell} (r)= \phi_{\ell}^{\lambda}(r)$ for each $\ell < i$. Therefore, by Lemma \ref{Sum},
$$\sum_{\ell \in \mathbb{Z}: \atop \ell < i} r_{\ell} - \sum_{\ell \in \mathbb{Z}: \atop \ell < i} \phi_{\ell}^{\lambda}(r) =\sum_{\ell \in \mathbb{Z}: \atop \ell < i} \prod_{k = \ell}^{i - 1} (1 - \lambda_k) r_{\ell}.$$
By \textit{consistency} and \textit{independence of future income streams},
\begin{eqnarray}
\phi^{}_{i} (r) &= & \phi^{}_{i} \biggl(\bigl( r_i +  \sum_{\ell \in \mathbb{Z}: \atop \ell < i} \prod_{k = \ell}^{i - 1} (1 - \lambda_k) r_{\ell} \bigr) e^i \biggr). \nonumber
\end{eqnarray}
Similarly,
$$\phi^{\lambda}_{i} (r) =  \phi^{\lambda}_{i} \biggl(\bigl( r_i +  \sum_{\ell \in \mathbb{Z}: \atop \ell < i} \prod_{k = \ell}^{i - 1} (1 - \lambda_k) r_{\ell} \bigr) e^i \biggr). $$
By Step 1, \textit{scale invariance} and the definition of $\phi^{\lambda}$,
\begin{eqnarray}
 \phi^{}_{i} \biggl(\bigl( r_i +  \sum_{\ell \in \mathbb{Z}: \atop \ell < i} \prod_{k = \ell}^{i - 1} (1 - \lambda_k) r_{\ell} \bigr) e^i \biggr) &= & \lambda_i \bigl( r_i +  \sum_{\ell \in \mathbb{Z}: \atop \ell < i} \prod_{k = \ell}^{i - 1} (1 - \lambda_k) r_{\ell} \bigr) \nonumber \\
&=& \phi_i^{\lambda} \biggl( \bigl( r_i +  \sum_{\ell \in \mathbb{Z}: \atop \ell < i} \prod_{k = \ell}^{i - 1} (1 - \lambda_k) r_{\ell} \bigr) e^i \biggr) \nonumber \\
 &=& \phi_{i}^{\lambda} (r), \nonumber 
 \end{eqnarray}
 as desired. 
\medskip

\noindent{\bf Step 3}. For each $r \in L^1_+$, $\phi(r) = \phi^{\lambda}(r)$.

Let $r \in L^1_+$ and construct the sequence $(r^m)_{m > 0}$ as follows: 
\begin{equation*}
r_j^m=\left\{ 
\begin{tabular}{cc}
$r_j$ & if $j \geq - m $ \\ 
$0$ & otherwise.%
\end{tabular}%
\right.
\end{equation*}
By Step 2, for each $m > 0$, $\phi(r^m) = \phi^{\lambda} (r^m).$
As $(r^m)_{m > 0}$ converges (with the $L^1$ norm) to $r$, it follows by \textit{continuity} of both $\phi^{}$ and $\phi^{\lambda}$, that  $\phi^{} (r) =  \phi^{\lambda} (r)$, which concludes the proof. 
\end{preuve}

Adding \textit{translation invariance} to the statement of Theorem \ref{charaResult}, one obtains a characterization of the subfamily of uniform geometric rules.
To see this, consider the infinite income streams $e^i$ and $e^{i+1}$ for some $i \in \mathbb{Z}$. By definition of $\phi^{\lambda}$ and \textit{translation invariance},
$$\phi_i^{\lambda} (e^i) = \lambda_i =  \phi_{i+1}^{\lambda} (e^{i+1}) = \lambda_{i +1}, $$
which forces $\lambda: \mathbb{Z} \longrightarrow [0, 1]$ to be a constant function.

\begin{corollary}
An allocation rule satisfies feasibility, scale invariance, independence of a future income, consistency, continuity and translation invariance if and only if it is a uniform geometric rule. 
Furthermore, it satisfies balance if and only if $\lambda \in (0, 1]$.
\end{corollary}

Adding \textit{idempotency} instead, 
one obtains a sharper characterization, as only the two extreme uniform geometric rules satisfy that axiom. 

\begin{corollary}
An allocation rule satisfies feasibility, scale invariance, independence of a future income, consistency, continuity and idempotency if and only if it is either the no-transfer rule or the full-transfer rule. 
Furthermore, it satisfies balance if and only if it is the no-transfer rule.
\end{corollary}

\section{Further insights}
Our previous results rely on a specific notion of continuity. We show in this section that natural alternative notions of continuity yield significantly different implications, highlighting the role of this principle in our analysis. This is in line with classical work in infinite-population ethics (e.g., Svensson, 1980). 
\medskip

We first consider the notion arising when replacing the taxicab norm by the sup-norm as follows. 
Let $(r^m)_{m \geq0}$ be a sequence within $L_+^1$. We say that $(r^m)_{m\geq 0}$ converges to $r \in  L_+^1$ in the sup-norm, and write $(r^m)_{m\geq 0}\overset{\infty}\to r$ if for each $\epsilon >0$, there exists $m_0 \in \N$ such that for each $m \geq m_0$, $\sup_{i \in \Z}\vert r^m_i  -r_i \vert  < \epsilon$.\\

\noindent {\bf Sup-continuity}: If $(r^m)_{m\geq 0}\overset{\infty}\to r$, then $(\phi(r^m))_{m \geq 0}\overset{\infty}\to \phi(r)$.\\

%

Note that sup-continuity is not logically related to continuity (albeit convergence in the taxicab norm implies convergence in the sup-norm). 
It turns out that not all geometric rules satisfy sup-continuity. To see this, consider the subfamily of geometric rules $\phi^{\lambda} \in {\cal G}^{\lambda}$ for which there exists $i < 0$ such that $\lambda_i = 0$, and such that $\lambda_0 = 1$. Define the sequence $(r^m)_{m\geq 0}$ as follows:
\begin{equation*}
r_i^m=\left\{ 
\begin{tabular}{cl}
$\frac{1}{m}$ & if $- m\leq i \leq - 1 $ \\ 
$0$ & otherwise.%
\end{tabular}%
\right.
\end{equation*}
Note that $(r^m)_{m \geq 0}\overset{\infty}\to0_{\mathbb{Z}}$. However, 
for each $m \geq 0$, $\phi_0^{\lambda}(r^m) = 1$, and, for each $i \in \mathbb{Z} \setminus 0$, $\phi_i^{\lambda} (r^m) = 0$. That is, $(\phi^{\lambda}(r^m))_{m \geq 0}$ is a constant sequence, different from $0_{\mathbb{Z}}$. Thus, $\phi^{\lambda}$ does not satisfy sup-continuity.
\medskip

The natural question arising from here is whether we can identify the subset of geometric rules satisfying sup-continuity. The next result answers positively such a question. 
For each geometric rule $\phi^{\lambda}$, and each $i \in \mathbb{Z}$, let
$$S^{\lambda}_i  =  \lambda_i  \sum_{j \in \mathbb{Z}: \atop j \leq i} \prod_{k = j}^{i - 1} ( 1 - \lambda_k).  $$
To interpret $S^{\lambda}_i$, note that $ \prod_{k = j}^{i - 1} ( 1 - \lambda_k)$ is the share of income $r_j$ transferred from $j$ to $i$. And that $\lambda_i$ refers to the share of it $i$ retains. 
Let ${\cal T}^{\lambda} \subset {\cal G}^{\lambda}$ denote the sub-family of geometric rules such that the set
$\{S^{\lambda}_i: i\in \mathbb{Z}\} $ is bounded from above.
 
\begin{theo} \label{charaResult2}
An allocation rule satisfies feasibility, scale invariance, independence of a future income, consistency and sup-continuity if and only if it belongs to ${\cal T}^{\lambda} \subset {\cal G}^{\lambda}$.
\end{theo}
\begin{preuve} 
By Theorem \ref{charaResult}, all rules within ${\cal T}^{\lambda}$ satisfy {\it feasibility}, {\it scale invariance}, {\it independence of a future income} and {\it consistency}. It remains to show that they satisfy {\it sup-continuity} and that no other rule within ${\cal G}^{\lambda}\setminus{\cal T}^{\lambda}$ does so.\\ 

\noindent {\bf Case 1}: $\phi^{\lambda} \in  {\cal T}^{\lambda}$. In this case, there exists $b \in \mathbb{R}$ such that for each $i \in \mathbb{Z}$, $S_i^{\lambda} \leq  b$. Now, for each pair $r, r' \in L_+^1$, and each $i \in \mathbb{Z}$,
\begin{eqnarray} \label{inequality}
|\phi_i^{\lambda}(r) - \phi_i^{\lambda} ( r')| &=&| \lambda_i \sum_{j \in \mathbb{Z}:\atop  j \leq i} \prod_{k = j}^{i - 1} ( 1 - \lambda_k) ( r_j - r_j' ) | \nonumber \\
                                                                   &\leq & \lambda_i \sum_{j \in \mathbb{Z}: \atop j \leq i} \prod_{k = j}^{i - 1} ( 1 - \lambda_k) | r_j - r_j'  | \nonumber \\
                                                                   & \leq & \lambda_i \sum_{j \in \mathbb{Z}: \atop j \leq i} \prod_{k = j}^{i - 1} ( 1 - \lambda_k) \sup_{j \in \mathbb{Z}} { |r_j - r_j'|}  \nonumber \\
                                                                   & \leq &  b \sup_{j \in \mathbb{Z}} { |r_j - r_j'|},  \nonumber 
\end{eqnarray}
from where it follows that $\phi^{\lambda}$ satisfies {\it sup-continuity}.\\ 

\noindent {\bf Case 2}: $\phi^{\lambda} \in   {\cal G}^{\lambda} \setminus  {\cal T}^{\lambda}$. In this case, for each $m \in \mathbb{N}$, there exists $i_m \in \mathbb{Z}$ such that $S_{i_m}^{\lambda} \geq m +1$. 
In particular, for each $m \in \mathbb{N}$, there exists $i_m\in \mathbb{Z}$ and $j_m \in \mathbb{Z}$  such that $S^{\lambda}_{i_m, j_m} \geq m$, where $S^{\lambda}_{i_m, j_m}$ is the partial sum of $S^{\lambda}_{i_m}$ over $\{j_m, \ldots, i_m\}$. That is,
$$ S^{\lambda}_{i_m, j_m} =  \lambda_{i_m}  \sum_{j \in \mathbb{Z}: \atop j_m \leq j \leq i_m} \prod_{k = j}^{i_m - 1} ( 1 - \lambda_k).$$
From this, construct the sequence $(r^m)_{m \in \mathbb{N}} \subseteq L_+^1$ as follows:
\begin{equation*}
r_i^m=\left\{ 
\begin{tabular}{cl}
$\frac{1}{m}$ & if $j_m\leq i \leq i_m $ \\ 
& \\
$0$ & otherwise.%
\end{tabular}%
\right.
\end{equation*}
It follows that $(r^m)_{m\geq 0}\overset{\infty}\to 0_{\mathbb{Z}}$. On the other hand,
\begin{eqnarray}
\phi^{\lambda}_{i_m} (r^m) &=&  \lambda_{i_m} \sum_{j \in \mathbb{Z}: \atop j_m \leq j \leq i_m} \prod_{k = j}^{i_m - 1} ( 1 - \lambda_k) \frac{1}{m} \nonumber \\
                   &= & S_{i_m, j_m}^{\lambda} \frac{1}{m}  \nonumber \\
                   &\geq& 1. \nonumber
\end{eqnarray}
Hence, for each $m \in \mathbb{N}$, 
$$ \sup_{i \in \mathbb{Z}}\{ \phi^{\lambda}_{i} (r^m) \}\geq 1,$$
from where it follows that $\phi^{\lambda}$ is not sup-continuous.\\


Conversely, let $\phi$ be an allocation rule satisfying {\it feasibility}, {\it scale invariance}, {\it independence of a future income}, {\it consistency} and {\it sup-continuity}. First,
note that the first item of Proposition \ref{relationaxiom} continues to hold under sup-continuity. Second, we proceed as in the uniqueness part of the proof of Theorem 1. In particular, note that the sequence constructed in Step 3 of that proof converges to $r$ in sup-norm and that $\phi^{\lambda}(r^m)$ converges in taxicab norm (and so in sup-norm too) to $\phi^{\lambda}(r)$. It follows that
that $\phi(r) = \phi^{\lambda}(r)$.\end{preuve}

As $\mathcal{U}^{\lambda} \subset {\cal T}^{\lambda}$, we have the following extra characterization results.

\begin{corollary} An allocation rule satisfies feasibility, scale invariance, independence of a future income, consistency, sup-continuity and translation invariance if and only if it is a uniform geometric rule. Furthermore, it satisfies balance if and only if $\lambda \in (0, 1]$.
\end{corollary}

\begin{corollary}
An allocation rule satisfies feasibility, scale invariance, independence of a future income, consistency, sup-continuity and idempotency if and only if it is either the no-transfer rule or the full-transfer rule. 
Furthermore, it satisfies balance if and only if it is the no-transfer rule.
\end{corollary}
\medskip

We conclude this section considering the continuity notion arising from point-wise convergence (again, with no logical relation with the other continuity axioms we have considered, albeit point-wise convergence implies convergence in the taxicab norm).\\

Let $(r^m)_{m \geq0}$ be a sequence within $L_+^1$. We say that $(r^m)_{m\geq 0}$ converges to $r \in  L_+^1$ point-wise and write $(r^m)_{m\geq 0}\overset{p}\to r$ if for each $\epsilon >0$, and for each $i \in \Z$, there exists $m_0 \in \N$ such that for each $m \geq m_0$, $\vert r^m_i  -r_i \vert  < \epsilon$.\\

\noindent {\bf Point-wise continuity}: If $(r^m)_{m\geq 0}\overset{p}\to r$, then $(\phi(r^m))_{m \geq 0}\overset{p}\to \phi(r)$.\\

%

Let ${\cal P}^{\lambda} \subset {\cal G}^{\lambda}$ denote the sub-family of geometric rules such that for each $i \in \mathbb{Z}$, there exists $j < i$, such that $\lambda_j = 1$. An interpretation of this condition is that transfers take place within disjoint blocks, as when $\lambda_j =1$ no income will be transferred to future generations. 
As the next result states, this set is precisely encompassing the geometric rules that satisfy point-wise continuity. 
  
  \begin{theo} \label{charaResult3}
An allocation rule satisfies feasibility, scale invariance, independence of a future income, consistency and point-wise continuity if and only if it is the full-transfer rule or it belongs to ${\cal P}^{\lambda}$.
\end{theo}

\begin{preuve}
By Theorem \ref{charaResult}, the full-transfer rule and all rules within ${\cal P}^{\lambda}$ satisfy {\it feasibility}, {\it scale invariance}, {\it independence of a future income} and {\it consistency}. It remains to show that they satisfy {\it point-wise continuity} and that no other geometric rule does so. 
\\

\noindent {\bf Case 1}:
 $\phi^{\lambda} \in {\cal P}^{\lambda}  \cup \{\phi^{0}\}$. 
If  $\phi^{\lambda} = \phi^{0}$, then it is trivially  point-wise continuous. Otherwise, by definition of ${\cal P}^{\lambda} $, for each $i \in \mathbb{Z}$,  there exists $j \in \mathbb{Z}$ such that   $ j< i $ and $\lambda_{j} = 1$.
Let $(r^m)_{m > 0} \overset{p}\to r$. 
Then, for each $m \in \mathbb{N}$, we have that for each $i \in \mathbb{Z}$, 
$$\phi^{\lambda}_i(r^m) - \phi^{\lambda}_i (r)  = 
        \lambda_i \sum_{j'  \in \mathbb{Z}:\atop   j< j' \leq i} \prod_{k = j'}^{j - 1} ( 1 - \lambda_k) ( r^m_{j'} - r_{j'}).$$ 
 Note that the above expression is a finite sum. Furthermore, for each $i \in \mathbb{Z}$ and each $j'\in \mathbb{Z}$ such that  $ j< j' \leq i$, $(r^m_{j'})_{m \in \mathbb{N}}\overset{p}\to r_{j'}$, which ensures that  $\phi^{\lambda}_i(r^m) - \phi^{\lambda}_i (r)\overset{p}\to 0$.  
Thus, $\phi^{\lambda}$ satisfies point-wise continuity. \\
 
\noindent {\bf Case 2}: $\phi^{\lambda} \in \mathcal{G}^{\lambda} \setminus ({\cal P}^{\lambda}  \cup \{\phi^{0}\})$. In this case, there exists   $i \in \mathbb{Z}$ such that $\lambda_i > 0$ and $\lambda_j < 1$ for each $j < i$. 
Consider the sequence $(r^m)_{m > 0}$ in $L^1_+$ such that:
$$r^m_{i - m} = \frac{1}{ \prod_{k = i - m}^{i -1} (1- \lambda_k)},$$
and $r^m_{i} = 0$ for each $i \in \mathbb{Z} \setminus \{i - m\}$. Then, $(r^m)_{m > 0}\overset{p}\to 0_{\mathbb{Z}}$.\footnote{Note that $(r^m)_{m > 0}\not\overset{\infty}\to 0_{\mathbb{Z}}$ as $r^m_{j - m}$ increases with $m$.} For each $m > 0$,
$$ \phi_i^{\lambda} (r^m) =  \lambda_i,$$
which implies that $\phi^{\lambda} (r^m)\not\overset{p}\to 0_{\mathbb{Z}}$, and so $\phi^{\lambda}$ violates point-wise continuity.\\

The converse implication follows from similar arguments to those in the uniqueness part of the proof of Theorem \ref{charaResult2}.
\end{preuve}

 Note that the only uniform geometric rules satisfying point-wise continuity are the full-transfer rule ($\lambda = 0$) and the no-transfer rule ($\lambda = 1$). This leads to the following characterization result.
 
 \begin{corollary}
An allocation rule satisfies feasibility, scale invariance, independence of a future income, consistency, point-wise continuity and translation invariance or idempotency if and only if it is either the no-transfer rule or the full-transfer rule. 
Furthermore, it satisfies balance if and only if it is the no-transfer rule.
\end{corollary}

To conclude, we order the sub-families of geometric rules introduced in this article. We know that
$$ {\cal G}^{\lambda} \supset  {\cal T}^{\lambda} \supset {\cal E}^{\lambda} \supset {\cal U}^{\lambda} \setminus \{\phi^0\} \,\,  \mbox{ and }  \,\, {\cal B}^{\lambda} \supset  {\cal E}^{\lambda} $$

\noindent $\bullet$ Let   $\phi^{\lambda} \in {\cal G}^{\lambda} $ where $\lambda: \mathbb{Z} \longrightarrow [0, 1]$ is defined as  

\begin{equation*}
\lambda_i=\left\{ 
\begin{tabular}{cl}
$1$ & if $i < 0$ \\ 
 $ 1$ & if $i = n(n +1)$ \mbox{ for }  $n \geq0$ \\
$0$ & otherwise.%
\end{tabular}%
\right.
\end{equation*}
By definition, $\phi^{\lambda} \in  {\cal P}^{\lambda} $. 
 On the other  hand, 
 $$S^\lambda_{n(n+1)} = n(n+1) -  n(n-1) = 2n,$$
 which shows that $\phi^{\lambda} \notin  {\cal T}^{\lambda}$. Note that every geometric rule  
 ${\cal U}^{\lambda} \setminus \{\phi^0\}$ belongs to ${\cal T}^{\lambda}$ but not to   ${\cal P}^{\lambda}$.
Finally, let  the geometric rules $\phi^{\lambda} \in {\cal G}^{\lambda} $ such that 
for each $i \in \Z $, $\lambda_{2i} = 1$ and $\lambda_{2i+1} = 1/2$.
It turns out $\phi^\lambda \in  {\cal P}^{\lambda} \cap {\cal T}^{\lambda}$. 

\noindent $\bullet$ Consider the geometric rule  $\phi^{\lambda} \in {\cal G}^{\lambda} $ where $\lambda: \mathbb{Z} \longrightarrow [0, 1]$ is defined as

\begin{equation*}
\lambda_i=\left\{ 
\begin{tabular}{cl}
$0$ & if $i > 0$ \\ 
 $ 1$ & if $i \leq 0$ \\
\end{tabular}%
\right.
\end{equation*}
Obviously, $\phi^{\lambda} \in {\cal T}^{\lambda} \cap  {\cal P}^{\lambda}$ but $\phi^{\lambda} \not \in {\cal B}^{\lambda}$. Next  take $\phi^{\lambda} \in {\cal G}^{\lambda} $ where $\lambda: \mathbb{Z} \longrightarrow [0, 1]$ is defined as
\begin{equation*}
\lambda_i=\left\{ 
\begin{tabular}{cl}
$1/2$ & if $i > 0$ \\ 
 $ 0$ & if $i \leq 0$ \\
\end{tabular}%
\right.
\end{equation*}
It is clear that $\phi^{\lambda} \in {\cal B}^{\lambda}$ but  $\phi^{\lambda} \not \in ({\cal T}^{\lambda} \cup  {\cal P}^{\lambda})$.

From the above, the following Venn diagram summarizes the relationships between the families of geometric rules, where
${\cal U}^{\star, \lambda}$ stands for ${\cal U}^{\lambda} \setminus \{ \phi^{0}\}$.
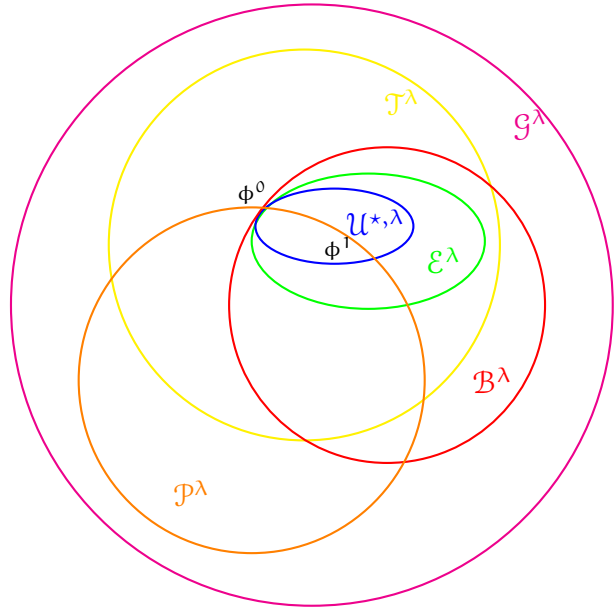
\bigskip

\begin{center}
 \begin{tikzpicture}[line cap=round,line join=round]

\coordinate (cG) at (0,0);      \def\rG{4.0}
\coordinate (cT) at (-0.1,0.8); \def\rT{2.6}
\coordinate (cB) at (1,0); \def\rB{2.1}
\coordinate (cE) at (-0.3,-0.0);  \def\rE{1.5}
\coordinate (cU) at (-0.3,-0.0);  \def\rE{1.5}

\coordinate (cP) at (-0.8,-1);\def\rP{2.3}

{\color{green}\draw[thick] (0.75, 0.85) ellipse (1.55 and 0.9);}
{\color{blue} \draw[thick] (0.3, 1.05) ellipse (1.05 and 0.5);}
{\color{magenta} \draw[thick]     (cG) circle[radius=\rG];}
{\color{yellow}\draw[thick]     (cT) circle[radius=\rT];}
{\color{red} \draw[thick]     (cB) circle[radius=\rB];}
{\color{orange} \draw[thick] (cP) circle[radius=\rP];}

\node at ($(cG)+(2.9,2.4)$) {{\color{magenta} ${\cal G}^{\lambda}$}};
\node at ($(cT)+(1.3,1.9)$) {{\color{yellow}${\cal T}^{\lambda}$}};
\node at ($(cB)+(1.4, -1)$) {{\color{red}${\cal B}^{\lambda}$}};
\node at ($(cE)+(2.05, 0.6)$) {{\color{green}${\cal E}^{\lambda}$}};
\node at ($(cP)+(-0.8,-1.5)$) {{\color{orange}${\cal P}^{\lambda}$}};
\node at ($(cU)+(1.15, 1.1)$) {{\color{blue}${\cal U}^{\star, \lambda}$}};

\node at (-0.8, 1.5) {\scriptsize $\phi^{0}$};
\node at (0.35, 0.75) {\scriptsize $\phi^{1}$};

\end{tikzpicture}
\end{center}

%
%
%
%
%
%

 \section{Discussion}
 In this paper, we have presented a stylized model for intergenerational transfers of income. We have characterized a family of rules that implement concatenated (geometric) transfers of income from past and present generations to future generations. Essentially, with geometric rules, each generation retains a share of the aggregate income received from past generations and transfers the remaining share to the next generation. We have shown that this family is characterized by the combination of feasibility, scale invariance, independence of a future income, consistency and continuity. For alternative continuity notions (replacing the standard taxicab norm by either the sup-norm or point-wise convergence), we characterize sub-families of geometric rules. 

As mentioned in the introduction, our work is somewhat related to the sizable literature on intergenerational equity, where the aim is to rank infinite utility streams (whereas in our case the aim is to provide rules that assign to each infinite income stream another infinite income stream, hence conveying intergenerational transfers of income). The main result in that literature (which also arises in the Arrovian framework with infinite population), is a tradeoff between the principles of unanimity (in the form of a weak Pareto axiom) and impartiality (in the form of an anonymity axiom). A typical way out from this tradeoff is to weaken the second principle. For example, by imposing finite anonymity rather than (full) anonymity. But this comes at the cost of the non-constructibility of the rules (e.g., Lauwers, 2010). 
In our analysis, no such tradeoff exists. We have shown that a combination of natural axioms characterizes a meaningful allocation rule. The only caveat is that the characterization depends on the notion of continuity. This is in line with \textit{the cost of continuity} in infinite population-ethics (see, for instance, Section 4.2 in Pivato and Fleurbaey (2024) and the references cited therein). 

A somewhat related literature is concerned with intertemporal choice (e.g., Chambers and Echenique, 2018; Feng and Ke, 2018; Billot and Qu, 2025). In that case, infinite streams are typically evaluated via discounted utilitarianism and the emphasis is on the appropriate choice of the discount factor.  
Chambers et al. (2023) characterize discount factors that display decreasing impatience through a convexity axiom for investments at fixed interest rates. They also show that they are precisely equivalent to a \textit{geometric} average of generalized quasi-hyperbolic discount rates. 

We should emphasize that a differential feature of our analysis is to work with a model in which both the past and the future are unbounded. That is, to represent the infinite (but countable) set of generations we consider $\mathbb{Z}$ rather than $\mathbb{N}$ (which is the standard in the literatures on intergenerational equity or intertemporal choice mentioned above). This happens to play an important role in our analysis when dealing with balanced rules and continuity.
More precisely, the balance condition $\prod_{k = i}^{\infty} (1 - \lambda_k)  = 0$ that defines the sub-family $\mathcal{B}^{\lambda}$, is a forward condition, as we start from $i$ and move towards the future. On the other hand, the conditions that define the sub-families $\mathcal{P}^{\lambda}$ and $\mathcal{T}^{\lambda}$ are backward conditions. In the former case, we say for each $i\in \mathbb{Z}$, there exists $j < i$ with $\lambda_j = 1$. In the latter case, we aggregate over all $j \leq i$. This partly explains why we can combine conditions to come up with intersections of different sub-families of geometric rules (as illustrated with the Venn diagram above).

To conclude, we also stress that our results can also be seen as extensions of counterpart results within the literature on fair allocation among finite populations (e.g., Moulin, 2014; Thomson, 2023). Geometric rules have been highlighted in various forms for related models within that field, which were listed at the introduction. Most of the principles we formalize in the axioms used in this paper also apply in those models. The distinctive feature comes from the role that continuity (in its various forms) plays in the model with infinite streams we study here.  

 \section{Appendix}

 \subsection{Proofs of the propositions}
 
In order to prove Proposition \ref{propBalance}, we need the following two intermediate results.

\begin{lemma} \label{sequence} Let $(x_i)_{i \geq 1}$ be a sequence of real numbers. Then, it holds that:
\begin{equation} \label{eqL1}
\forall k\geq 1, \quad \prod_{i = 1}^k (1 - x_i) + \sum_{i = 1}^{k} x_i  \prod_{j = 1}^{i - 1} (1 - x_j) = 1. 
\end{equation} 
\end{lemma}
\begin{preuve} The proof is by induction on $k \geq 1$.\medskip

{\sc Base Case}: $k = 1$. In this case, (\ref{eqL1}) becomes $(1 - x_1) + x_1 = 1$, which obviously holds. \medskip

{\sc Induction hypothesis}: Assume that  (\ref{eqL1}) holds for $k \geq 1$. \medskip

{\sc Induction step}: Consider the rank $k+1$. One obtains
\begin{eqnarray}
\prod_{i = 1}^{k+1} (1 - x_i) + \sum_{i = 1}^{k+1} x_i  \prod_{j = 1}^{i - 1} (1 - x_j) & = & (1 - x_{k+1})\prod_{i = 1}^k (1 - x_i) + x_{k +1} \prod_{j = 1}^{k} (1 - x_j) + \sum_{i = 1}^{k} x_i  \prod_{j = 1}^{i - 1} (1 - x_j) \nonumber \\
& = & \prod_{i = 1}^k (1 - x_i) + \sum_{i = 1}^{k} x_i  \prod_{j = 1}^{i - 1} (1 - x_j) \nonumber \\
& = & 1, \nonumber
\end{eqnarray}
where the last equality follows from the induction hypothesis.  
\end{preuve}
Note that, as (\ref{eqL1}) holds for each $k \geq 1$, it also follows that 
\begin{equation} \label{eqL2}
\prod_{i = 1}^{+\infty} (1 - x_i) + \sum_{i = 1}^{+\infty} x_i  \prod_{j = 1}^{i - 1} (1 - x_j) = 1. 
\end{equation}

\begin{lemma} \label{Sum} For each 
  $\phi^{\lambda} \in {\cal G}^{\lambda}$, the following two statements hold:
\begin{enumerate}
\item $$\forall j \in \mathbb{Z}, \quad \sum_{i \in \mathbb{Z}:\atop i < j}\phi_i^{\lambda} (r) = \sum_{i \in \mathbb{Z}: \atop i < j} \bigl(1 - \prod_{k = i}^{j- 1} (1 - \lambda_k)\bigr)r_i;$$
\item  $$\sum_{i \in \mathbb{Z}}\phi_i^{\lambda} (r) = \sum_{i \in \mathbb{Z}} \bigl(1 - \prod_{k = i}^{+ \infty} (1 - \lambda_k)\bigr)r_i.$$
\end{enumerate} 
\end{lemma}

\begin{preuve} For the first statement, let $j \in \mathbb{Z}$. By definition of a geometric rule,
\begin{eqnarray}
\sum_{i \in \mathbb{Z}:\atop i < j} \phi_i^{\lambda} (r) &=&  \sum_{i \in \mathbb{Z}:\atop i < j} \sum_{\ell \in \mathbb{Z}:\atop \ell \leq i}\lambda_i\prod_{k =\ell }^{i - 1} ( 1 - \lambda_k)r_{\ell} \nonumber  \\
                                                                                    &= & \sum_{\ell \in \mathbb{Z}:\atop \ell < j} r_{\ell}  \sum_{i \in \mathbb{Z}:\atop \ell \leq i < j}\lambda_i \prod_{k =\ell }^{i - 1} ( 1 - \lambda_k)\nonumber
\end{eqnarray} 
By (\ref{eqL1})  in Lemma \ref{sequence},
$$ \sum_{i = \ell}^{j -1}\lambda_i \prod_{k =\ell }^{i - 1} ( 1 - \lambda_k) = 1 - \prod_{i = \ell}^{j - 1} ( 1 - \lambda_i),$$ 
from which one deduces that
\begin{eqnarray}
\sum_{i \in \mathbb{Z}:\atop i < j} \phi_i^{\lambda} (r) &=&  \sum_{\ell \in \mathbb{Z}:\atop \ell < j} \bigl(1 - \prod_{i = \ell}^{j - 1} ( 1 - \lambda_i)\bigr)r_{\ell} \nonumber \\
                                                                                    &=&\sum_{i \in \mathbb{Z}:\atop i < j}\bigl(1 - \prod_{k = i}^{j - 1} ( 1 - \lambda_k)\bigr)r_{i}, \nonumber
 \end{eqnarray}
as desired.  \medskip 

As for the second statement, one just needs to proceed in a similar way by using (\ref{eqL2}).                                                                        
\end{preuve}

\begin{preuve} (of Proposition \ref{propBalance}). Let $r \in L_+^1$. By the second statement of Lemma \ref{Sum},
$$\sum_{i \in \mathbb{Z}}\phi_i^{\lambda} (r) = \sum_{i \in \mathbb{Z}} \bigl(1 - \prod_{k = i}^{+ \infty} (1 - \lambda_k)\bigr)r_i \leq \sum_{i \in \mathbb{Z}} r_i,$$
where the inequality follows from the fact that $1 - \lambda_k \in [0, 1]$ for each $k \in \mathbb{Z}$. This proves that $\phi^{\lambda}(r)$ is a feasible allocation. It also follows from here that the allocation is balanced if and only if (\ref{c-balance}) holds.
\end{preuve}

\begin{preuve} (of Proposition \ref{Lip-continuous}). Let $\phi^{\lambda} \in {\cal G}^{\lambda}$ 
  and $r, r' \in L^1_+$. We show that 
$$\sum_{i \in\mathbb{Z}} |\phi_i^{\lambda}(r) -  \phi_i^{\lambda}(r')| \leq \sum_{i \in\mathbb{Z}} | r_i - r_i'|,$$
from where continuity follows. 
By definition of a geometric rule,
\begin{eqnarray} \label{TI-Lip}
\sum_{i \in\mathbb{Z}} |\phi_i^{\lambda}(r) -  \phi_i^{\lambda}(r')| & = & \sum_{i \in\mathbb{Z}} |\lambda_i \sum_{j \in \mathbb{Z}: \atop j \leq i} \prod_{k = j}^{i - 1} (1 - \lambda_k)r_j - \lambda_i \sum_{j \in \mathbb{Z}: \atop j \leq i} \prod_{k = j}^{i - 1} (1 - \lambda_k)r'_j | \nonumber \\
&= &\sum_{i \in\mathbb{Z}} \lambda_i |\sum_{j \in \mathbb{Z}: \atop j \leq i} \prod_{k = j}^{i - 1} (1 - \lambda_k) (r_j - r_j')| \nonumber \\
&\leq & \sum_{i \in\mathbb{Z}} \lambda_i \sum_{j \in \mathbb{Z}: \atop j \leq i} \prod_{k = j}^{i - 1} (1 - \lambda_k) |r_j - r_j'|, 
\end{eqnarray}
where the last inequality follows from the triangular inequality. \medskip

Now, for each $i \in \mathbb{Z}$, let $z_i = |r_i - r_i'|$. Then, $z \in L^1_+$. As $\phi^{\lambda}$ satisfies \textit{feasibility}, it follows by (\ref{TI-Lip}) that
\begin{eqnarray}
\sum_{i \in\mathbb{Z}} |\phi_i^{\lambda}(r) -  \phi_i^{\lambda}(r')| &\leq&  \sum_{i \in\mathbb{Z}} \lambda_i \sum_{j \in \mathbb{Z}: \atop j \leq i} \prod_{k = j}^{i - 1} (1 - \lambda_k) |r_j - r_j'| \nonumber \\
&=& \sum_{i \in\mathbb{Z}} \phi_i^{\lambda}(z) \nonumber \\
&\leq & \sum_{i \in\mathbb{Z}} z_i \nonumber \\
&=& \sum_{i \in\mathbb{Z}} | r_i - r_i'|, \nonumber
\end{eqnarray}
which proves the result.
\end{preuve}

\begin{preuve} (of Proposition \ref{relationaxiom}). For the first statement, let $\phi$ be an allocation rule satisfying \textit{independence of a future income} and \textit{continuity}. Let $r, r' \in L_+^1$ and $j \in \mathbb{Z}$ be such that $r_k =r'_k$ for each $k < j$. 
For each $m\in \mathbb{N}$, we construct the sequence $(r^m)_{m \geq 0}$ in $L^1_+$ as follows:
\begin{equation*}
r_k^m=\left\{ 
\begin{tabular}{cc}
$r'_k$ & if $k < m+j $ \\ 
$r_k$ & otherwise.%
\end{tabular}%
\right.
\end{equation*}
By \textit{continuity}, $\phi(r^m) \longrightarrow  \phi (r')$. 
By successive applications of \textit{independence of a future income}, for each $i < j$, 
$$\quad \phi_i(r) = \phi_i(r^0) = \ldots = \phi_i(r^m),$$
which forces $\phi_i(r)=  \phi_i (r')$. This shows that $\phi$ satisfies \textit{independence of future income streams}.\medskip

As for the second statement, let $\phi$ be an allocation rule satisfying \textit{independence of future income streams} and \textit{feasibility}. Let $j \in \mathbb{Z}$ and $r \in L_+^1$ be such that $r_i = 0$ for each $i \leq j$. By \textit{independence of future income streams}, 
$\phi_j(r) = \phi_j(0_{\mathbb{Z}})$.
By \textit{feasibility},
$$\sum_{i \in \mathbb{Z}} \phi_i(0_{\mathbb{Z}}) \leq 0,$$
which forces $\phi_i(0_{\mathbb{Z}}) = 0$ for any $i \in \mathbb{Z}$, and so $\phi_j(r) = 0$, as desired.
\end{preuve}
 \subsection{Logical independence of the axioms in Theorem \ref{charaResult}}
 
 \begin{itemize}
\item Consider the allocation rule $\phi$ on $L_+^1$ defined as $\phi_{i} (r) = 2r_i$ for each $i \in \mathbb{Z}$. It violates feasibility, 
 but satisfies the remaining axioms in Theorem \ref{charaResult}.
 \item Define the allocation $\phi$ on $L_+^1$ as follows:
 \begin{itemize}
 \item $\phi_0(r) = r_0$ if $r_0 \leq r_1$, and $\phi_0(r) = r_1$ otherwise;
 \item $\phi_1(r) = r_1$  if $r_0 \leq r_1$,  and $\phi_1(r) = r_0$ otherwise;
 \item $\phi_1(r) = r_i$ for each other $i \in \mathbb{Z} \setminus \{0, 1\}$.
 \end{itemize}
  It violates independence of a future income, but satisfies all the other axioms of Theorem  \ref{charaResult}.
 \item The allocation rule $\phi$ on $L_+^1$ defined as:
 \begin{itemize}
 \item $\phi_0 (r) =  \frac{r_0^2}{ 1 + r_0}$;
 \item $\phi_1 (r) = r_1 + r_0 -  \frac{r_0^2}{ 1 + r_0}$;
  \item $\phi_i (r) = r_i$ for each other $i \in \mathbb{Z} \setminus \{0, 1\}$,
   \end{itemize}
  violates scale invariance, 
  but satisfies all the other axioms of Theorem  \ref{charaResult}.
  \item Consider the allocation $\phi$ on $L_+^1$  defined as:
    \begin{itemize}
    \item $\phi(r) = \phi^{\lambda}(r)$ for some function $\lambda = (1/2, \ldots,1/2, \ldots)$ if $r$ has a finite support, that is if the set $\{r_i \not = 0 : i \in \mathbb{Z}\}$ is finite, and
    \item $\phi(r) = r$ otherwise.
     \end{itemize} 
     This allocation rule is not continuous. To see this, consider the sequence $(r^m)_{m \geq 0} \subseteq L_1^+$, where 
     $$r_i^m = \frac{1}{2^i} \quad \mbox{if} \quad 0 \leq i \leq m, \quad \mbox { and } \quad  r_i^m = 0 \quad \mbox{otherwise}.$$
 Then, $(r^m)_{m\geq 0}\overset{1}\to r$, where $r$ is defined as $r_i = 1/2^i$ for $i \geq 0$ and $r_i = 0$ otherwise.
 Now, $\phi_0 (r^m) = 1$ for each $m\geq 0$ and $\phi_0 (r) = 1/2$. Thus, $\phi(r^m)_{m\geq 0}\not\overset{1}\to \phi(r)$. 

 On the other hand, this allocation rule satisfies all the remaining axioms in the statement. 
 \item The allocation rule $\phi$ on $L_+^1$ defined as:
 
 $$\forall i \in \mathbb{Z}, \quad \phi_i (r) = \frac{r_i}{4} +  \frac{r_{i - 1}}{4}+  \frac{r_{i - 2}}{4},$$
 violates consistency, but satisfies all the remaining axioms in Theorem \ref{charaResult}. 
  \end{itemize}
 
 \bigskip 


\begin{thebibliography}{99}
\bibitem{} Ambec, S., Sprumont, Y. (2002). Sharing a river. \textit{Journal of Economic Theory} 107, 453-462.
\bibitem{} Asheim, G. (1991). Unjust intergenerational allocations. \textit{Journal of Economic Theory} 54, 350-371.
\bibitem{} Asheim, G. (2010). Intergenerational Equity. \textit{Annual Review of Economics} 2, 197-222.
\bibitem{} Asheim, G., Bossert, W., Sprumont, Y., Suzumura, K. (2010). Infinite-horizon choice functions. \textit{Economic Theory} 43, 1-21.
\bibitem{} Billot, A., Qu, X. (2025). Stationary altruism and time consistency. \emph{Journal of Economic Theory} 228, 106038.
\bibitem{BG} van den Brink, R., Gilles, R. (1996). Axiomatizations of the conjunctive permission value for games with permission structures. \textit{Games and Economic Behavior} 12, 113-126.
\bibitem{} Brown, D., Lewis, L. (1981). Myopic economic agents. \emph{Econometrica} 49, 359-368.
\bibitem{Casajus15} Casajus, A. (2015). Monotonic redistribution of performance-based allocations: A case for proportional taxation. \emph{Theoretical Economics} 10, 887-892.
\bibitem{} Chambers, C., Echenique, F. (2018). On multiple discount rates. \emph{Econometrica} 86, 1325-1346.
\bibitem{} Chambers, C., Echenique, F., Miller, A. (2023). Decreasing impatience. \emph{American Economic Journal: Microeconomics} 15, 527-551.
\bibitem{Chambers2021} Chambers, C., Moreno-Ternero, J. (2021). Bilateral redistribution. \emph{Journal of Mathematical Economics} 96, 102517.
\bibitem{} Diamond, P. (1965). The evaluation of infinite utility streams. \textit{Econometrica} 33, 170-177.
\bibitem{} Dietzenbacher, B., Kondratev, A. (2023). Fair and consistent prize allocation in competitions. \textit{Management Science} 69, 3319-3339.
\bibitem{} Dietzenbacher, B., Tamura, Y., Thomson, W. (2024). Partial-implementation invariance and claims problems. \textit{Social Choice and Welfare} 63, 203-229.
\bibitem{} Epstein, L. (1986). Intergenerational consumption rules: an axiomatization of utilitarianism and egalitarianism. \textit{Journal of Economic Theory} 38, 280-297.
\bibitem{} Feng, T. Ke, S. (2018). Social discounting and intergenerational Pareto. \emph{Econometrica} 86, 1537-1567.
\bibitem{} Gudmundsson, J., Hougaard, J., Moreno-Ternero, J., \O sterdal, L. (2025). Optimizing successive incentives: rewarding the past or motivating the future? \textit{Games and Economic Behavior} 153, 10-29.
\bibitem{} Harless, P. (2020). Robust revenue sharing on a hierarchy. \textit{Economic Theory}, in press.
\bibitem{} Hougaard, J., Moreno-Ternero, J., Tvede, M., \O sterdal, L. (2017). Sharing the proceeds from a hierarchical venture. \textit{Games and Economic Behavior} 102, 98-110.
\bibitem{} Hougaard, J., Moreno-Ternero, J., \O sterdal, L. (2022). Optimal management of evolving hierarchies. \textit{Management Science} 68, 6024-6038.
\bibitem{} Ju, B-G. (2013). Coalitional manipulations on networks. \textit{Journal of Economic Theory} 148, 627-662.
\bibitem{} Ju, B-G., Jeong, S., Song, H., Jo, K. (2024). Revenue sharing on hierarchies. Center for Distributive Justice Discussion Paper 2024-03. Seoul National University.  
\bibitem{} Ju, B-G., Miyagawa, E., Sakai, T. (2007). Non-manipulable division rules in claim problems and generalizations. \textit{Journal of Economic Theory} 132, 1-26.
\bibitem{} Lauwers, L. (2010). Ordering infinite utility streams comes at the cost of a non-Ramsey set. \textit{Journal of Mathematical Economics} 46, 32-37.
\bibitem{} Kondratev, A., Ianovski, E., Nesterov, A. 2024. How should we score athletes and candidates: geometric scoring rules. \textit{Operations Research} 72, 2507-2525.
\bibitem{} Koopmans, T. (1960). Stationary ordinal utility and impatience. \textit{Econometrica} 28, 287-309.
\bibitem{} Mart\'{\i}nez, R., Moreno-Ternero J. (2025). The allocation of riparian water rights. \textit{Environmental and Resource Economics} 88, 3841-3872. 
\bibitem{} Mart\'{\i}nez, R., Moreno-Ternero J. (2026). The geometric adjudication of water rights in international rivers. ArXiv 2601.04150.
\bibitem{} Moreno-Ternero J., Roemer J. (2006). Impartiality, solidarity, and priority in the theory of justice. \textit{Econometrica} 74, 1419-1427.
\bibitem{} Moreno-Ternero, J. Roemer, J. (2008). The veil of ignorance violates priority. \textit{Economics and Philosophy} 24, 233-257.
\bibitem{} Moulin, H. (2004). Fair division and collective welfare. MIT Press.
\bibitem{ONeill1982} O'Neill, B. (1982). A problem of rights arbitration from the Talmud. \emph{Mathematical Social Sciences} 2, 345-371.
\bibitem{} Piacquadio, P. (2014). Intergenerational egalitarianism. \textit{Journal of Economic Theory} 153, 117-127.
\bibitem{} Pivato, M., Fleurbaey, M. (2024). Intergenerational equity and infinite-population ethics: A survey. \textit{Journal of Mathematical Economics} 113, 103021.
\bibitem{} Rawls, J. (1971). A Theory of Justice. Harvard University Press. Cambridge, MA.
\bibitem{} Roemer J. (1986). Equality of resources implies equality of welfare. \textit{Quarterly Journal Economics} 101, 751-784.
\bibitem{} Svensson, L.-G. (1980). Equity among generations. \emph{Econometrica} 1251-1256.
\bibitem{} Thomson, W. (2012). On the axiomatics of resource allocation: interpreting the consistency principle. \textit{Economics and Philosophy} 28, 385-421.
\bibitem{} Thomson, W., (2023). The Axiomatics of Economic Design, Vol. 1. Studies in Choice and Welfare. Springer, Cham.
\bibitem{Young1988} Young, H.~P. (1988). Distributive justice in taxation. \emph{Journal of Economic Theory} 44, 321-335.

\end{thebibliography}
\end{document}